\documentclass[sigplan,natbib=false,nonacm]{acmart}
\renewcommand\footnotetextcopyrightpermission[1]{}
\acmYear{2023}
\copyrightyear{2023}

\RequirePackage{algorithm}
\RequirePackage[noend]{algpseudocode}
\algdef{SL}[REPLICASTATE]{ReplicaState}{0}[1]{\textbf{state} #1}
\algdef{SL}[INITIAL]{Initial}{0}[1]{\textbf{initial} #1}
\algdef{SL}[PRE]{Pre}{0}[1]{\textbf{pre} #1}
\algdef{SL}[AWAIT]{Await}{0}[1]{\textbf{await} #1}
\algdef{SL}[ASSERT]{Assert}{0}[1]{\textbf{assert} #1}
\algblockdefx[FN]{Fn}{EndFn}%
[3]{\textbf{function} \emph{#1} (#2) : #3}{}
\algblockdefx[QUERY]{Query}{EndQuery}%
[3]{\textbf{query} \emph{#1} (#2) : #3}{}
\algblockdefx[UPDATE]{Update}{EndUpdate}%
[2]{\textbf{update} \emph{#1} (#2)}{}
\algblockdefx[GENERATE]{Generate}{EndGenerate}%
[1]{\textbf{generate} (#1)}{}
\algblockdefx[EFFECT]{Effect}{EndEffect}%
[1]{\textbf{effect} (#1)}{}
\algtext*{EndUpdate}
\algtext*{EndQuery}
\algtext*{EndGenerate}
\algtext*{EndEffect}
\algtext*{EndFn}
\DeclareMathOperator{\mlb}{mlb}
\newcommand{\mathvar}[1]{\mathit{#1}}
\newcommand{\mathval}[1]{\mathtt{#1}}
\newcommand{\mathfn}[1]{\mathop{\mathvar{#1}}}

\RequirePackage{hyperref}
\RequirePackage[capitalise]{cleveref}
\newtheorem{theorem}{Theorem}
\newtheorem{lemma}{Lemma}
\renewcommand*{\proofname}{Proof Sketch}

\RequirePackage[
  datamodel=acmdatamodel,
  style=acmnumeric,
  ]{biblatex}
\begin{document}
\title[On Extend-Only Directed Posets and Derived Byzantine-Tolerant Replicated Data Types (Extended Version)]{On Extend-Only Directed Posets and Derived Byzantine-Tolerant Replicated Data Types {\Large(Extended Version)}}

\author{Florian Jacob}
\email{florian.jacob@kit.edu}
\orcid{0000-0002-5739-8852}
\affiliation{%
  \institution{Karlsruhe Institute of Technology}
  \streetaddress{Kaiserstraße 12}
  \city{Karlsruhe}
  \country{Germany}
  \postcode{76131}
}

\author{Hannes Hartenstein}
\email{hannes.hartenstein@kit.edu}
\orcid{0000-0003-3441-3180}
\affiliation{%
  \institution{Karlsruhe Institute of Technology}
  \streetaddress{Kaiserstraße 12}
  \city{Karlsruhe}
  \country{Germany}
  \postcode{76131}
}

\begin{abstract}
    We uncover the extend-only directed posets (EDP) structure as a unification of recently discussed DAG-based Byzantine-tolerant conflict-free replicated data types (CRDT).
    We also show how a key-value map model
    can be derived from the EDP formulation, and give an outlook on an EDP-based systemic access control CRDT 
    as a formalization of the CRDT used in the Matrix messaging system.
\end{abstract}



\maketitle

\section{Introduction}
Recently, we noticed \cite{JacobBeerHenze2021_1000129939} that the conflict-free replicated data type (CRDT) of a system for decentralized messaging (Matrix, \cite{matrix-spec}) retains its CRDT property also in environments with Byzantine nodes, i.e., in environments with nodes that arbitrarily deviate from the expected protocol behavior.
\citeauthor{kleppmann2022making}~\cite{kleppmann2022making} was able to show that an arbitrary crash-fault tolerant CRDT can be transformed into a CRDT that tolerates an arbitrary number of Byzantine nodes.
These proposals for the Byzantine case have in common that events are managed in a graph-like shared object: events are independently generated at the various replicas (available under partition), but associated to previous events known to the replica.

The shared objects essentially show the following dynamic behavior.
All replicas start with the same genesis event as initial state.
To append a new event, replicas attach this event to all events currently known to them that “happened-before” and have no “descendants”, i.e., the “most recent” events, without coordinating with other replicas.
Replicas then synchronize, i.e., exchange updates, to reach a consistent state again.
When replicas append events that happened concurrently, branches occur, which leads to a tree-like structure.
When replicas learn of branches, they will join them again on their next event, thereby eliminating branches.
Events consist of both payload as well as hashes of previous events, which ensures integrity and strong eventual consistency in the face of Byzantine equivocation~\cite{kleppmann2022making,JacobBayreutherHartenstein2022_1000146481}.
An example evolution over time of such an object is illustrated in \cref{fig:graph}.

While in previous work, the corresponding objects have been described as directed acyclic graphs (DAGs), in this work-in-progress paper we propose a more fundamental, unifying CRDT formalization via {\em directed partially-ordered sets} (directed posets), called Extend-only Directed Posets (EDPs), that shows the following advantages:
\begin{itemize}
  \item The set-theoretic formalization captures the essence of Byzantine-tolerant DAG CRDTs,
    revealing {\it a)} a state-based EDP that does not require crypto elements,
    and {\it b)} a hash-based operation-based construction as optimization for efficiency.
    In contrast to DAG-based variants, the EDP formulation makes use of standard mathematical notions
    and fits the theory of CRDTs that is largely formulated in terms of set theory~\cite{preguicca2018conflict}.
    Due to the close relation of set theory and especially lattice theory to boolean algebra~\cite{padmanabhan2008axioms}, we believe that EDPs might facilitate formal verification in future work.
  \item Based on the generic EDP formalization, other types can be derived and/or constructed by composition that can then more easily be shown to have the CRDT property in Byzantine environment.
    We give an example for a derived map type, and an outlook towards a composed CRDT for systemic access control, which represents a formalization of the CRDT used in Matrix.
\end{itemize}


\begin{figure}[tbp]
\centering
\includegraphics[trim=6mm 86mm 174mm 51mm,clip,width=\linewidth]{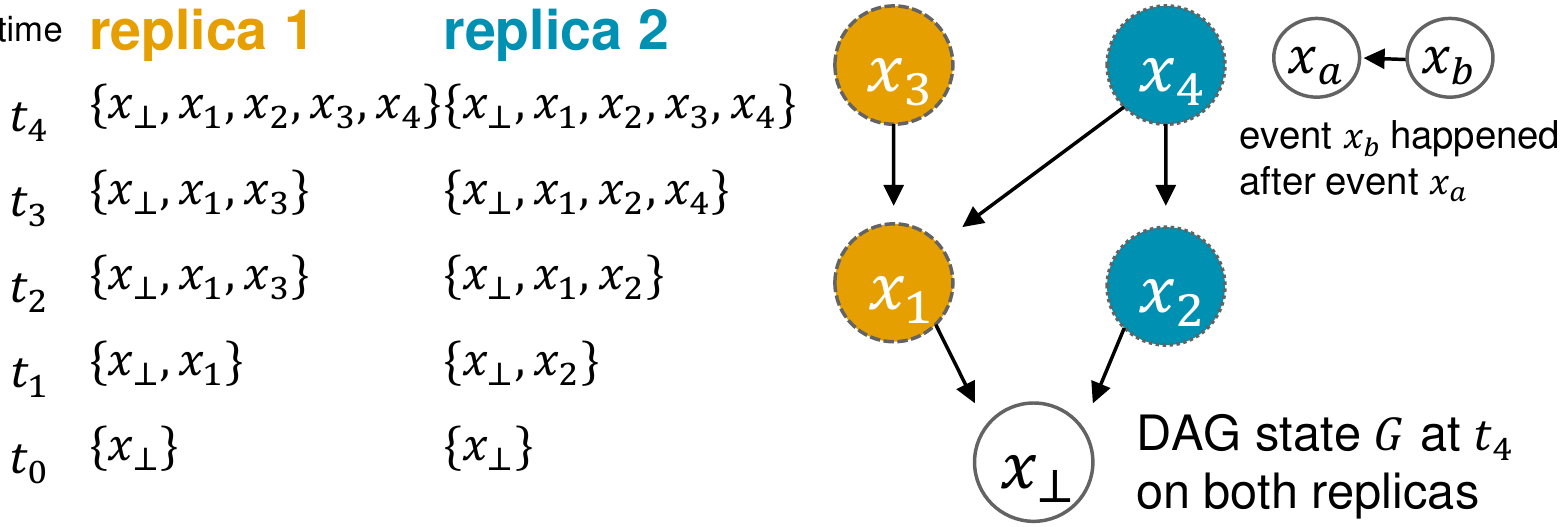}
\caption{Example graph state $G$ at a point in time $t_4$ as seen by both replicas and containing events from both replicas and the pre-shared genesis event $x_{\bot}$. First, both replicas send events $x_1$ and $x_2$ in concurrently. When replica 2 sends event $x_4$, it has learned about $x_1$ from replica 1 and therefore puts both $x_1$ and $x_2$ as happened-after parents.
Replica 1 has not learned of $x_2$ yet, and thereby appends $x_3$ only to $x_1$.}
\label{fig:graph}
\vskip-5mm
\end{figure}


Related work made use of sets instead of DAGs in Byzantine environments, however, did not treat and generalize them as CRDTs~\cite{cholvi2021byzantine}.
In the following sections, we will define EDPs as CRDTs and indicate why they represent a solid basis for the class of Byzantine-tolerant CRDTs discussed above.


\section{Extend-only Directed Posets (EDPs)}

\subsection{Basics and Notations}

The state space of the EDP replicated data type is based on relational structures, for which we will now introduce mathematical basics and notation.

A relational structure $S = (X, R)$ is a tuple which consists of a ground set $X$ and a relation $R \subseteq X \times X$.
To facilitate notation, we sometimes write $S.X$ and $S.R$ to denote the set $X$ and the relation $R$ of $S=(X,R)$.
A partially-ordered set, also called a poset, is a relational structure that is reflexive ($\forall a \in X\colon (a, a) \in R$), transitive ($\forall a, b, c \in X\colon (c, b) \in R \land (b, a) \in R \Rightarrow (c, a) \in R$), and antisymmetric ($\forall a, b \in X\colon (a, b) \in R \land (b, a) \in R) \Rightarrow a = b$).
While a $\leq$-like $R$ is usual in mathematics, we use a $\geq$-like $R$ in these definitions for consistency with the rest of the paper.
If $S$ is a poset, then $R$ is called a partial order relation, or partial ordering.
If $S$ is also strongly connected ($\forall a, b \in X\colon (a, b) \in R \lor (b, a) \in R$), $S$ is called a linearly-ordered set.


A reflexive and transitive relational structure $S$ is called a downward-directed set if for any two elements, the set contains a lower bound, i.e., $\forall a, b \in X \colon \exists x_{lb} \in X \;\mbox{s.t.}\; (a, x_{lb}) \in R \land (b, x_{lb}) \in R$.
A finite, downward-directed poset is directed towards its unique least element $x_\bot$, also denoted as $S^\bot$, i.e., $x_\bot \in X$ and $\forall x \in X\colon (x, x_\bot) \in R$.
Conversely, a finite upward-directed poset is directed towards its top element $S^\top = x_\top \in X$, i.e., $ \forall x \in X \colon (x_\top, x) \in R$.
A both downward- and upward-directed finite poset is said to be a poset bounded by $x_\bot$ and $x_\top$.
The set of maximal elements of $S$ is $\max(S) = \{ m \in X | \forall x \in X \colon (x, m) \in R \Rightarrow (m, x) \in R \}$.
Conversely, the set of minimal elements of $S$ is $\min(S) = \{ m \in X | \forall x \in X \colon (m, x) \in R \Rightarrow (x, m) \in R \}$.

A relational structure $S'=(X',R') $ is called an extension of another relational structure $S$ if $X \subseteq X'$ and $ R = R'|_X$, where the restriction $R|_A$ is defined as usual as $R|_A =  R \cap (A \times A)$.
We call $S'$ an upward extension of a downward-directed poset $S$ if additionally $S'^\bot = S^\bot$.
The downward closure $y^{\downarrow S}$ of an element $y \in S.X$ is defined as $y^{\downarrow S} = \{c \in X | (y, c) \in R\}$.
The downward closure $Y^{\downarrow S}$ of a subset $Y \subseteq X$ is generalized from the single-element downward closure as $Y^{\downarrow S} = \bigcup_{y \in Y} y^{\downarrow S}$.
The upward closure of elements and subsets of $X$ is defined correspondingly, $y^{\uparrow S} = \{c \in X | (c, y) \in R\}$.
The set of maximal lower bounds of an element $y \in X$ is the set of maximal elements of the downward closure of $y$ without $y$ itself, $\mlb(y) = \max(y^{\downarrow S} \setminus \{y\})$.

\begin{figure}[tbp]
\centering
\includegraphics[trim=25mm 31mm 162mm 61mm,clip,width=\linewidth]{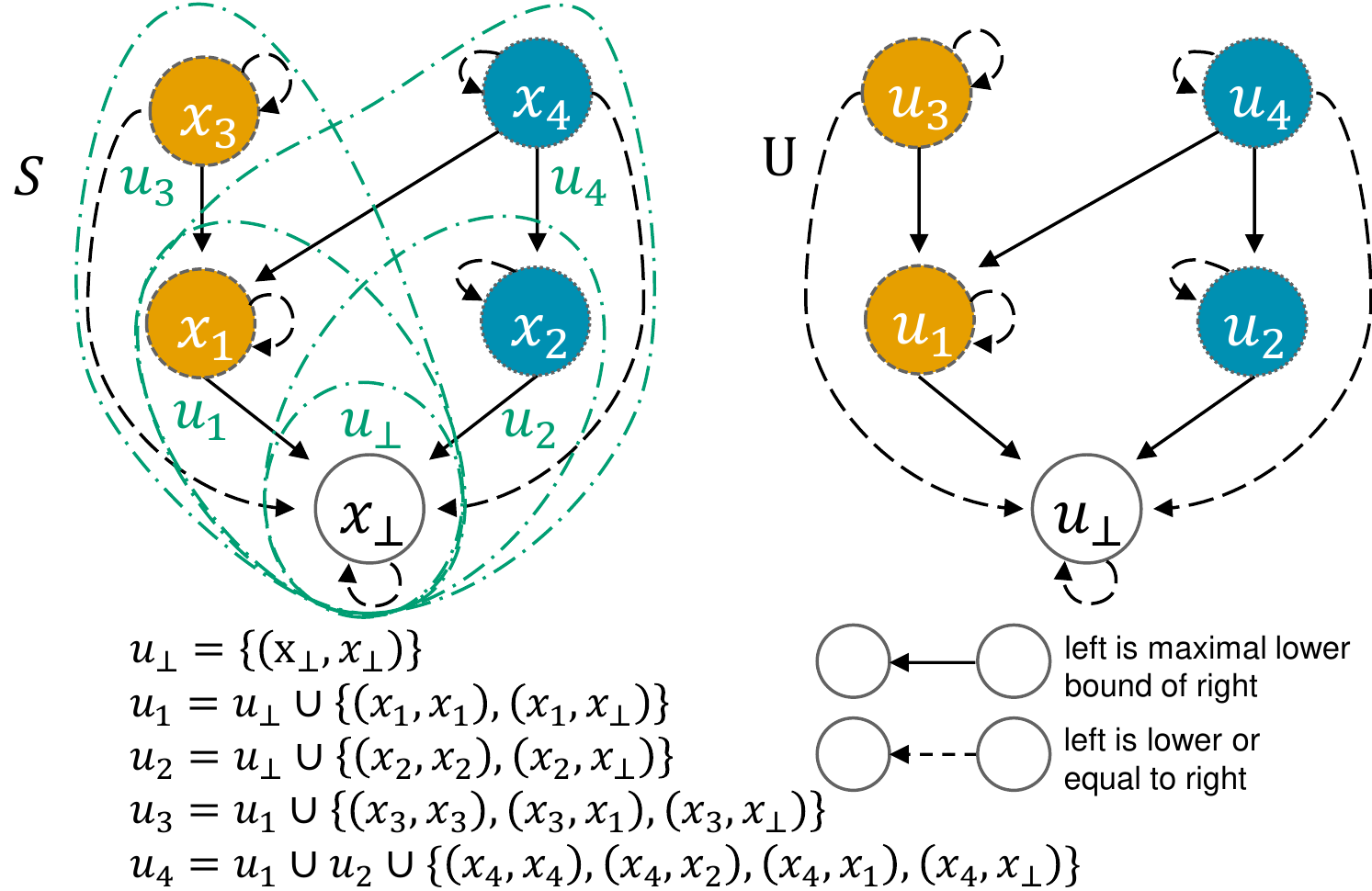}
\caption{A BDP state $S$ and equivalent EDP state $U$ based on \cref{fig:graph}.
An upward extension $u_i \in U$ is the set of its enclosed relations in $S$, to represent the formation history of $x_i \in S$.
}
\label{fig:homomorphism}
\end{figure}

\subsection{Specification}
\label{sec:edp-spec}
To build an append-only CRDT (as sketched in the introduction) that tolerates an arbitrary number of Byzantine nodes, the key idea is to replace an event $x$ with its downward closure, i.e., adding all the relations to previous events.
Thereby, a byzantine node cannot create inconsistent relations anymore --- it can only create spam additions.
The EDP RDT will be defined, therefore,  as `higher-order' directed posets, i.e., will represent directed posets of directed posets.
The resulting EDP RDT can then easily be shown to represent a state-based CRDT, for which an operation-based version can be constructed by using cryptographic hash functions as order-guaranteeing mechanisms.
For clarity, we denote as Basic Directed Posets (BDP) the directed posets that are used as `base layer' to construct the EDP RDT.

To fix a BDP, one selects a universe $\mathbb{X}$ of valid elements  and a universe $\mathbb{R} \subseteq \mathcal{P}(\mathbb{X}^2)$ of valid relations.
BDP states are finite, downward-directed posets $S = (X, R)$, $X \subseteq \mathbb{X}$, $R \in \mathbb{R}$,
with a common bottom element $S^\bot$ (to be interpreted as the initial state of all correct replicas).
As example, one can think of the universe $\mathbb{X}$ as application-layer messages, and the universe $\mathbb{R} \subseteq \mathcal{P}(\mathbb{X}^2)$ as causal orders, and a specific bootstrapping application-layer message as bottom element.

To define extend-only operations, we only want to allow single-element upward extensions of a BDP state $S$.
We want finality as required validity criterion for extensions: Adding a new element to the BDP and adding the relations of the new element needs to be a single, atomic upward extension, i.e., cannot be changed later.
The required finality property can be expressed as follows: when element $x$ is the single element that has been added to gain state $S$, for all upward extensions $S'$ of $S$ needs to hold\footnotemark{}: $x^{\downarrow S'} = x^{\downarrow S}$ and $R'|_{x^{\downarrow S'}} = R|_{x^{\downarrow S}}$.
\footnotetext{In contrast, the upward closure of an element and their relations might be never final: $x^{\uparrow S'} \supseteq x^{\uparrow S}$ and $R'|_{x^{\uparrow S'}} \supseteq R|_{x^{\uparrow S}}$. A new upper element can always be in transit or purported to have been in transit.}

To provide this finality notion, the obvious (but inefficient) approach is to bind a single-element upward extension, which extends $S = (X, R)$ with an element $x \in \mathbb{X}$, to all single-element upward extensions for elements $y$ with $x \geq y$.
This way, the history of extend-only operations is fixed and serves as identity-forming information for the new single-element upward extension.
To express sets of `formation histories' of BDP elements, one has to move to sets of posets.
Thus, for the EDP definition, we move `one level up' and map those BDP posets to elements of the EDP as done in the following two steps.
An illustration of the BDP to EDP state mapping based on the state in \cref{fig:graph} is found in \cref{fig:homomorphism}.

{\em Step 1.} Let the directed poset $S'=(X',R')$ be the upward extension with a single element $x \in \mathbb{X}$ of a directed poset $S=(X,R)$ of the BDP.
The downward closure $x^{\downarrow S'}$ with $R'$ restricted to $x^{\downarrow S'}$ is a sub-poset $(x^{\downarrow S'}, R'|_{x^{\downarrow S'}}) \subseteq S'$ bounded by $x_\bot$ and x.
Let $u_x := R'|_{x^{\downarrow S'}}$ denote this relation of such a BDP upward extension with $x$, and $X(u_x) := \{y \in \mathbb{X} | (y, y) \in u_x\} = x^{\downarrow S'}$ the set of reflexive pairs in $u_x$.
We can derive from $u_x$ the upward extension $S'$ of $S$ with $x$ as $S' = (X \cup X(u_x), R \cup u_x)$.
Thereby, as shorthand notation, we call $u_x$ an upward extension as well.
An upward extension $u_x \in \mathbb{R}$ is valid if $(X(u_x), u_x)$ forms a BDP sub-poset of $S'$ bounded by $x_\bot$ and~$x$.
While we focus on its relation, the bounded sub-poset is an alternative upward extension shorthand.

{\em Step 2.} We now move `one level up' and define the EDP by using upward extensions $u_x$ as elements and subset relations to form posets of those upward extensions.
An EDP state $U \in \mathcal{P}(\mathbb{R})$ is the set of single-element upward extensions $U = \{u_\bot, u_1, \mathellipsis\}$, i.e., the formation history of BDP state $S$.
The initial EDP state is $U = \{ u_\bot = \{ (x_\bot, x_\bot) \} \}$, which is the upward extension of the empty set with the bottom element.
An EDP state $U \in \mathcal{P}(\mathbb{R})$ is valid if $(U, \supseteq|_U)$ is a $u_\bot$-directed poset and $\forall u \in U\colon \mlb(u) \subsetneq U \land |X(u)| = |\{ X(\bigcup \mlb(u))| + 1$, i.e., every upward extension in $U$ also has its maximal lower bounds in $U$, and extends the BDP state with exactly one element that is not present in any of its maximal lower bounds.
\cref{fig:set} provides an illustration of applying an upward extension, derived from the graph of \cref{fig:graph}.
An EDP state $U$ can be transformed back to a directed poset $S(U)$ on the underlying BDP via $S(U) = (X(\bigcup U), \bigcup U)$.
The above EDP definition leads to a {\em state-based CRDT}
with state space $\mathcal{P}(\mathbb{R})$ and set union $U_1 \cup U_2$ as join, as shown in~\cref{th:state-edp-sec}.

\begin{theorem}
\label{th:state-edp-sec}
Assuming a connected component of all correct replicas and eventual communication among the component,
the state-based EDP is a Conflict-free Replicated Data Type
even in face of an arbitrary number of Byzantine replicas.
\end{theorem}
\begin{proof}
For the state space $\mathcal{P}(\mathbb{R})$, set union $U_1 \cup U_2$ is the least upper bound of $U_1, U_2 \in \mathcal{P}(\mathbb{R})$, whereby $\mathcal{P}(\mathbb{R})$ and set union form a join-semilattice.
Via periodic state gossiping, eventual delivery is fulfilled.
Together with termination from the purely mathematical state join, the state-based EDP is a CRDT~\cite{shapiro2011conflict}.
As updates only consist of a semilattice element, there is no metadata to forge, and Byzantine replicas are limited:
Invalid Byzantine updates are rejected as not part of the semilattice, and due to the full formation history protecting the integrity of valid updates, any Byzantine attempt to alter history, equivocate, or otherwise harm consistency, is equivalent to multiple valid updates which were just not successfully sent to all replicas~\cite{JacobBayreutherHartenstein2022_1000146481}.
\end{proof}

While the modeling as directed posets might introduce some formalism, the above formulation of the EDP CRDT uncovers the simple set-theoretic fundamentals, i.e., the state-based structure, of Byzantine-tolerant DAG CRDTs.
Based on this set-theoretic view, \cref{th:state-edp-sec} is easily shown without needing crypto elements like hashes for Byzantine tolerance, as the full relation among set elements is sufficient for integrity.
However, from a practical point of view, a state-based EDP is highly inefficient: replicas continuously gossip their full state $U$, which will only increase in size.
To reach the efficiency of DAGs with EDPs, we now work up our way from the state-based formulation to an efficient operation-based EDP CRDT formulation that makes use of cryptographic hashes for integrity protection.
The construction follows \cite{JacobBeerHenze2021_1000129939,kleppmann2022making}, its purpose here is to clarify the relationship between operation-based and state-based formalizations.
We need two optimizations for which we give an intuition now, and a formalization next.
Resilient operation broadcast follows in~\cref{sec:broadcast}, and a proof sketch for Byzantine Strong Eventual Consistency follows in~\cref{sec:edp-sec}.
As first optimization, an upward extension $u_y$ represents a sub-poset of $S'$ bounded by $x_\bot$ and $y$, thus, sending only the relations of that sub-poset instead of the state $U'$ is sufficient.
For the \emph{operation-based formulation}, we then compress the subset, reducing worst-case update size from “depth” to “width” of $U'$:
We define an operation $o$ consisting of $y$ and the set of hashed maximal lower bounds of $u_y$ in $U'$.
As $\mlb(u_y) \subseteq U$, $u_y$ can be reconstructed based on operation $o$ and existing state $U$.

\begin{figure}[tbp]
\centering
\includegraphics[trim=25mm 62mm 152mm 40mm,clip,width=.9\linewidth]{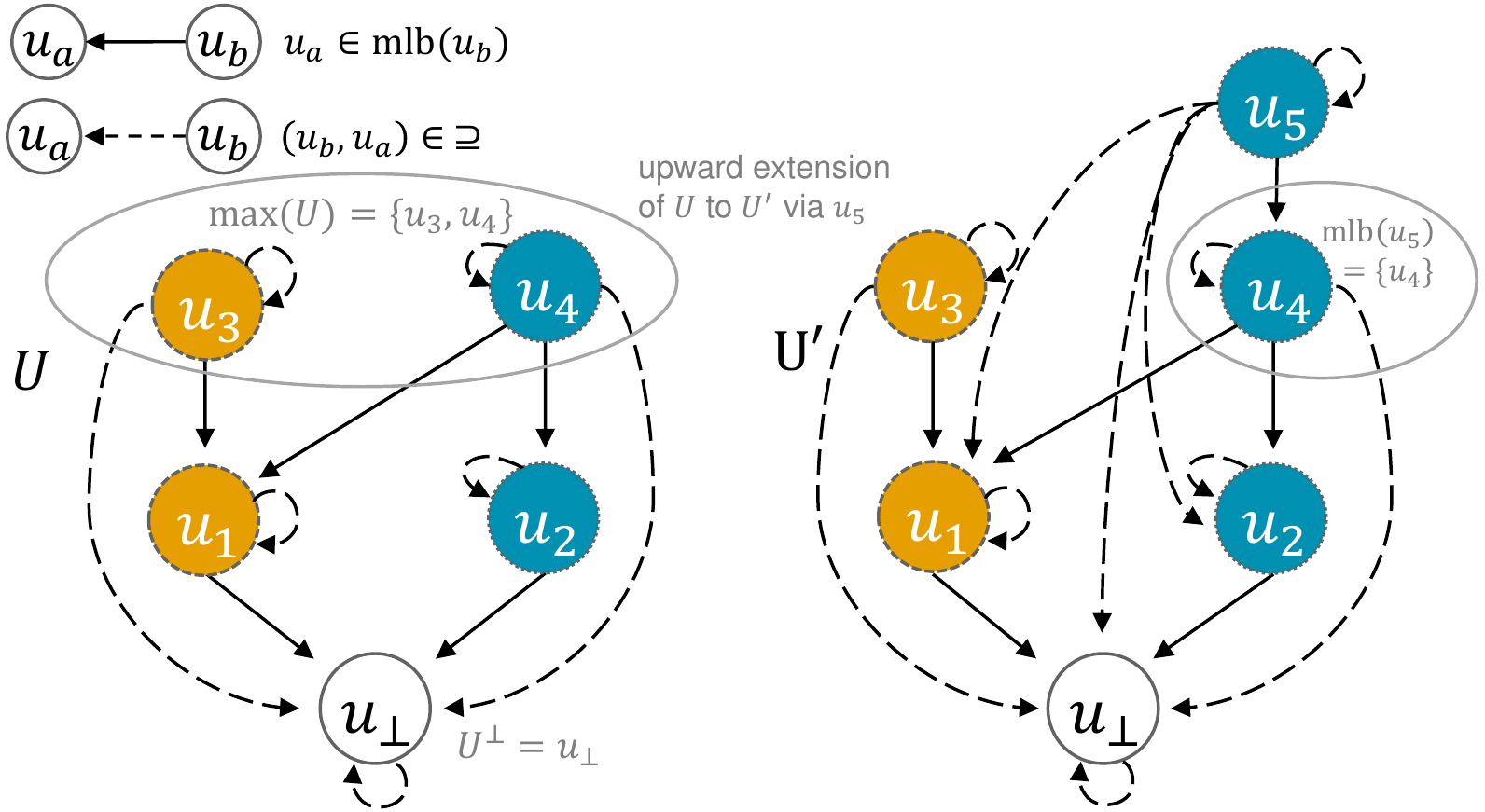}
\caption{Example of a replica state $U$ of an EDP, and state $U'$ that results from an upward extension of $U$ with $u_5$.}
\label{fig:set}
\end{figure}

To chain $y$ to its maximal lower bounds, a preimage and collision resistant hash function $h(a, b)$ is used that returns the hash of its chained arguments.
We recursively define the set of hashes of a set of maximal lower bounds $U'_{\mlb}$ as $H(U'_{\mlb}) := \{ h(\, \max(X(u)),\, H(\mlb(u)) \,) | u \in U'_{\mlb} \}$.
To define the hash of the set of hashes in the second argument of $h$, one can concatenate its linearization gained by lexicographical sorting.
An operation $o:=(y, H_{\mlb})$ is then defined as compression of upward extension $u$ via $o(u) := (\, \max(X(u)),\, H(\mlb(u)) \,)$.
The hashed  maximal lower bound set recursively protects the integrity of the directed sub-poset down to the bottom element $u_\bot$ against Byzantine nodes.
To reconstruct upward extension $u$ from operation $o$ based on knowledge of current state $U$, we need the set of relations of all of $o.y$'s predecessors $P := (\bigcup U'_{\mlb} \subseteq U | H(U'_{\mlb}) = o.H_{mlb})$.
Then, the upward extension is the union of the reflexive relation of $y$, the relation of $y$ to its predecessor elements $X(P)$, and its transitive relation, i.e., its predecessor relations $P$.
We formalize this reconstruction as $u(y, P) := \{ y \}^2 \cup \{ y \} \times X(P) \cup P$.
An operation $o$ can be locally applied as soon as all maximal lower bounds are part of the replica state, i.e., $o.H_{\mlb} \subseteq H(U)$.
An operation $o$ is applied to state $U$ by $U' = U \cup u(o.y, P)$,
inheriting commutativity from set union, which makes for an operation-based CRDT with state space $\mathcal{P}(\mathbb{R})$.
Update size is now bound to the maximum size of a maximal lower bound set of any element in the downward closure, instead of overall state size.

The operation-based EDP is given in~\cref{alg:edprdt}.
Encountering an unsatisfied assertion, the algorithm stops processing the function and returns an error.
Encountering an unsatisfied await, the algorithm interrupts to await its potential future satisfaction, without blocking subsequent function calls.
The \texttt{extend} function is used to extend the current state $U$ with a new upward extension $u \in \mathbb{R}$.
The side-effect-free \texttt{generate} function generates update operation $o$ and broadcast it to all replicas, including itself.
Received broadcasts are processed by the \texttt{effect} function, which awaits all maximal lower bounds 
to be part of the current state before applying the operation.
Byzantine tolerance is sketched in \cref{sec:edp-sec}, and requires a weak resilient broadcast outlined next.

\begin{algorithm}[t]
\caption{Op-Based EDP CRDT for BDP $(X \subseteq \mathbb{X}, R \in \mathbb{R})$}
\label{alg:edprdt}
\begin{algorithmic}
\Require universe of valid BDP elements $\mathbb{X}$
\Require universe of valid BDP relations $\mathbb{R} \subseteq \mathcal{P}(\mathbb{X}^2)$
\ReplicaState{set of upward extensions $U \in \mathcal{P}(\mathbb{R})$}
\Initial{$U \gets \{ u_\bot = \{ (x_\bot, x_\bot) \} \}$}

\Query{bot}{}{$u_\bot = U^\bot \in U$}
\EndQuery

\Query{max}{}{$U_{\max} = \max(U) \subseteq U$}
\EndQuery

\Query{mlb}{$u \in U$}{$U_{\mlb} = \mlb(u) \subsetneq U$}
\EndQuery

\Query{S}{}{directed poset $S = (X \subseteq \mathbb{X}, R \in \mathbb{R})$}
  \State $S \gets (X(\bigcup U\}, \bigcup U)$
\EndQuery

\Statex
\Update{extend}{$u \in \mathbb{R}$}
  \Generate{$u \in \mathbb{R}$}
    \Assert{$\emptyset \neq u \notin U$}
    \Assert{$\emptyset \neq \mlb(u) \subsetneq U$}
    \State $o \gets (\max(X(u)), H(\mlb(u)))$
  \EndGenerate
  \Effect{$o = (y \in \mathbb{X}, H_{\mlb})$}
    \Await{$H_{\mlb} \subseteq H(U)$} \Comment{await effect of updates $\subsetneq u$}
    \Assert{$H_{\mlb} \neq \emptyset$}
    \State $P \gets \bigcup (U'_{\mlb} \subseteq U | H(U'_{\mlb}) = H_{mlb})$
    \State $U \gets U \cup (\{ y \}^2 \cup \{ y \} \times X(P) \cup P)$
  \EndEffect
\EndUpdate

\end{algorithmic}
\end{algorithm}

\subsection{Resilient Broadcast of Operation-Based Updates}
\label{sec:broadcast}
The operation-based EDP formulation does not need the strong guarantees of a crash-/Byzantine-fault reliable broadcast.
As explained in~\cite{kleppmann2022making}, mere eventual delivery of updates is sufficient for such Byzantine Sybil-resistant CRDTs, 
as long as 
correct replicas form a connected component.
Due to the shared bottom element and relation directedness,
broadcasting the set of maximal elements is sufficient,
as missing elements can be iteratively queried from other replicas and integrity-verified via the hash chain.
An optimized broadcasting approach in this spirit is found in~\cite{kleppmann2022making}.

When the set of maximal upward extensions $\hat{U} = \max(U)$ of the replicas' current state $U$ (or their space-efficient operation equivalents) is gossiped regularly, the broadcast is a state-based CRDT itself:
The state space forms a join-semilattice with join of $\hat{U}_1$ and $\hat{U}_2$ being $\max(\hat{U}_1 \cup \hat{U}_2)$.
While update size is close to EDP size in worst case,
if all replicas are correct, update size converges quickly to approximately the number of involved replicas~\cite{JacobBeerHenze2021_1000129939}.
Byzantine replicas can send a large set of made-up extensions, but as the resulting state-based broadcast is a Byzantine Sybil-resistant CRDT~\cite{JacobBayreutherHartenstein2022_1000146481}, they can only attack performance but not correctness.

\subsection{Op-based Byzantine Strong Eventual Consistency}
\label{sec:edp-sec}
To show Byzantine strong eventual consistency of the op-based EDP RDT, we use the strong eventual consistency (SEC) notion as defined in~\cite{kleppmann2020byzantine} consisting of three properties:
a) \textbf{Self-update}, i.e., iff a correct replica generates an update, it applies that update to its own state; b) \textbf{Eventual update}, i.e., for any update applied by a correct replica, all correct replicas will eventually apply that update; and c) \textbf{Strong Convergence}, i.e., any two correct replicas that have applied the same set of updates are in the same state.

\begin{lemma}
\label{lm:directedness}
Let $U \in \mathcal{P}(\mathbb{R})$ be a directed poset corresponding to $S = (X \subseteq \mathbb{X}, R \in \mathbb{R})$, and $U'$ ($S'$ resp.) the resulting state after applying the update $o = (y, H_{\mlb})$ for upward extension $u$.
Then $U'$ ($S'$ resp.) are a partially-ordered extensions of $U$ ($S$ resp.) directed towards the same element $U^\bot$ ($S^\bot$ resp.).
\end{lemma}
\begin{proof}
Invalid upward extensions are discarded via assertions, and then $U' = U$.
Valid upward extensions only add a single element $y$ to $S$.
Due to $P$ being the union of $u_\bot$-directed sub-posets and $\{ y \} \times X(P) \subsetneq u$, $U'^\bot = U^\bot$ and also $S'^\bot = S^\bot$, i.e., $\bot$-directedness is kept.
Due to the same argument, the partial-order properties also still hold for both $U'$ and $S'$.
\end{proof}

\begin{theorem}
\label{th:edp-sec}
Under the assumption of a connected component of all correct replicas and eventual communication among the component,
the op-based EDP is a Conflict-free Replicated Data Type
even in face of an arbitrary number of Byzantine replicas.
\end{theorem}
\begin{proof}
\textbf{Self-update:}
With a broadcast as described in \cref{sec:broadcast}, the broadcasting replica receives the update without waiting for any acknowledgment.
The \texttt{generate} function creates an update for which the await precondition in \texttt{effect} is immediately satisfied, which means that the \texttt{effect} function directly updates the replica's own state.
Byzantine nodes have no attack vector to interfere with this process.
\textbf{Eventual update:} Using a broadcast as described in \cref{sec:broadcast}, as soon as one correct replica receives an update, eventually every correct replica will receive the update.
Due to eventual communication among the connected component of correct replicas, Byzantine nodes have no attack vector to interfere with eventual delivery of correct updates among correct replicas.
Through the hash-chained maximal lower bounds, replicas can verify the integrity and completeness of the directed sub-poset of the downward closure up until the bottom element on every new update, without a Byzantine node being able to interfere.
As soon as one correct replica can satisfy the await precondition for applying an update, i.e., has received the necessary downward closure for the element to add, it will eventually share the downward closure with all correct replicas, for which the delivery precondition is then also fulfilled eventually so that they can proceed with the effect function as well.
\textbf{Strong convergence:}
Due to \cref{lm:directedness} and the commutativity of set union, all valid updates commute.
Due to the \texttt{effect} function's assertions, only valid updates are applied.
Due to an operation consisting of both its payload as well as the hashes of its maximal lower bounds, the integrity of the directed poset of the downward closure of an element is integrity-protected through hash chaining.
Thereby, a Byzantine node that tries to attack consistency via equivocation with two operations $(y_a, H^a_{\mlb})$ and $(y_b, H^b_{\mlb})$ where either $y_a = y_b$ or $H^a_{\mlb} = H^b_{\mlb}$ is not successful, as all correct replicas will reject neither update as already received, but treat them as separate updates.
As the validity checks are deterministic and the same on all replicas, Byzantine replicas cannot get an update applied on only a proper subset of correct replicas.
\end{proof}

Under the given assumption of eventual communication among correct replicas, neither the operation-based nor the state-based EDP require authenticity of updates on their own terms, i.e., to provide SEC.
As their update functions do not depend on the identity of replicas that created an update, a forged but valid update is just \emph{another} valid update.
However, digital signatures are required as soon as the application requires set elements to contain an identifier for the sending replica, e.g., for access control, which we will discuss in~\cref{sec:ac-outlook}.

\subsection{Causal Extend-Only Directed Poset (CEDPs)}
We denote the subtype of Extend-only Directed Posets with temporal events as set items together with their causal relation, in form of their happened-after-equal relation, as Causal Extend-only Directed Posets (CEDPs).
As an EDP subtype, \cref{th:edp-sec} also applies.
Hash chaining can represent any upward-extend-only partial ordering, but inherently proves that the image happened-after the preimage, and thereby
the happened-after relation in a Byzantine-tolerant way.

CEDPs are the set-based version, i.e., the reflexive-transitive closure, of causal Event DAG approaches like the Matrix Event Graph (MEG)~\cite{matrix-spec}.
A specific proof that the MEG is a Byzantine-tolerant CRDT is found in~\cite{JacobBeerHenze2021_1000129939}.
In~\cite{kleppmann2020byzantine}, an Event DAG approach is used as transport layer to make arbitrary CRDTs Byzantine-tolerant.
Their arguments for Byzantine tolerance and broadcast requirements also apply to CEDPs.

While the causal set approach to discrete, logical time is well-known in quantum physics~\cite{bombelli1987space}, in distributed systems, the happened-before relation as defined by Lamport~\cite{lamport2019time} is common.
The happened-after-equal relation is the converse of the reflexive closure of the happened-before relation, and thereby the relations are interchangeable.
The happened-after-equal relation however represents the direction of hash chaining and an ordering-based notion of equality.
In Physics, an event is defined as a point in space-time.
An operation $o = (x, H_{\mlb})$ can be read as “event $x$ happened at discrete logical time coordinate $H_{\mlb}$ and discrete space coordinate of the (implicit) replica identifier”.
The happened-after order is a superset of the causal order, i.e., the order in which events causally depend on each other.
Either the application provides new events and their maximal lower bounds in causal order, or we use the happened-after order via $\max(U)$ as maximal lower bounds, which also covers the actual causal order.

\section{Replicated Data Types Derived from EDPs}
\subsection{EDP-based Maps (EPMs)}
\label{sec:epm}
\subsubsection{EPM Specification}
\label{sec:epm-spec}
Based on the EDP replicated data type, we derive a map replicated data type we call EPM.
Essentially, the universe $\mathbb{X}$ now represents key-value pairs, and as before the relation $\mathbb{R}$ defines how update operations are ordered.
The EPM has a `largest-element-wins' semantics with respect to an ordering of update operations based on $\mathbb{R}$.

Formally, we define $\mathbb{M} \subseteq \mathbb{X}$ as the set of valid key-value pairs of the form $(k \mapsto v)$, and a map $M \subseteq \mathbb{M}$ as `injective' subset of all valid key-value pairs, i.e., a given key only has one unique associated value.
We define the square bracket operator to query keys for map $M$ as $(k \mapsto v) \in M \Leftrightarrow M[k] = v$ and a map update operator $\uplus$ for single-element upward extensions $u$ that keeps injectiveness of~$M$:
\begin{align*}
\forall u \in \mathbb{R}&, x = \max(X(u)): \\ M \uplus u
:= & \begin{cases}
M \setminus \{k \mapsto \_\} \cup \{k \mapsto v\} & \text{if } x = (k \mapsto v) \in \mathbb{M}\\
M & \text{if } x \in \mathbb{X} \setminus \mathbb{M}
  \end{cases}
\end{align*}

The functions of the EPM replicated data type are given in \cref{alg:epmrdt}.
The core of the algorithm is the \emph{linearize} function, which linearizes a set $T \subseteq U$ partially-ordered by $\subseteq$ to a sequence $T_n$.
We define a relation $R_\parallel$ of all pairs of upward extensions $(u_1, u_2) \in T^2$ that are `parallel', i.e., cannot be compared using $\subseteq$, and are also minimal in the sense that no elements smaller than $u_1$ exist that cannot be compared to $u_2$, and vice-versa for $u_2$.
Using a preimage and collision resistant cryptographic hash function $h$, we define a strict linear order relation $R_h = \{ (u_1, u_2) \in \mathbb{R}^2 | h(u_1) < h(u_2) \}$.
The relation $R_\parallel \cap R_h$ then contains the necessary tie-breakings for the partial ordering $\subseteq$ to gain the linear ordering $R_l = (\subseteq \cup (R_\parallel \cap R_h))^+$.
The hash function allows resolving ties without Byzantine nodes being able to compromise the outcome.

\begin{algorithm}[b]
\caption{Operation-Based EPM Replicated Data Type}
\label{alg:epmrdt}
\begin{algorithmic}

\ReplicaState{set of upward extensions $U \in \mathcal{P}(\mathbb{R})$}
\Initial{$U \gets \{ u_\bot = \{ (x_\bot, x_\bot) \} \}$}
\Statex

\Update{put}{$(k \mapsto v) \in \mathbb{M}$}
  \State $y \gets (k \mapsto v)$
  \State $\mathfn{extend}(u(y, \bigcup \max(U)))$ \Comment{function of \cref{alg:edprdt}}
\EndUpdate

\Fn{linearize}{$T \subseteq U$}{$T_n \in \mathbb{R}^n$}
  \vspace{-.3em}
  \begin{align*}
    R_\parallel \gets \{&(u_1, u_2) \in T^2 | u_1 \nsubseteq u_2 \land u_2 \nsubseteq u_1\\
    & \land \forall u \subsetneq u_1: u \subsetneq u_2 \land \forall u \subsetneq u_2: u \subsetneq u_1 \}
    \end{align*}
  \State $R_l \gets (\subseteq \cup (R_\parallel \cap R_h))^+$ \Comment{order $R_\parallel$ via $R_h$,}\\
  \Comment{$^+$ denotes the reflexive-transitive closure}
  \State $T_n \gets \mathfn{enumerate}(T, R_l)$ \Comment{set $T \to$ sequence $T_n$}
\EndFn

\Query{get}{$T \subseteq U$}{$M \in \mathbb{M}$}
  \State $T_n \gets \mathfn{linearize}(T^{\downarrow U})$
  \State $M \gets T_0 \uplus T_1 \uplus \mathellipsis \uplus T_n$
\EndQuery

\end{algorithmic}
\end{algorithm}

Note that Matrix also provides maps based on the CEDP similar to \cref{sec:epm} to assign additional attributes to replicas and replicated objects, and uses them to provide access control, which we will discuss in~\cref{sec:ac-outlook}.

\subsubsection{Byzantine Eventual Consistency Verification}
\begin{theorem}
  Under the assumption of a connected component of all correct replicas and eventual communication among the component, an EPM is an op-based Conflict-free Replicated Data Type
  even in face of an arbitrary number of Byzantine replicas.
\end{theorem}
\begin{proof}
By reduction to the Byzantine strong eventual consistency of the underlying op-based EDP.
In a correct replica, the key-value pair is put in context through the set of maximal elements $\max(U)$ as maximal lower bounds of $(k \mapsto v)$, which satisfies the await-precondition and keeps the \textbf{Self-Update property}.
Nothing has changed in the broadcasting and application of updates, and the \texttt{linearize} function does not interact with other replicas but takes all updates applied in the EDP into account, which keeps \textbf{Eventual Delivery}.
The \texttt{get} and \texttt{linearize} functions as well as the map update operator $\uplus$ are deterministic and ignore invalid updates, so given the same downward-directed poset $(U, \supseteq)$, i.e. the same state of the EDP, they return the same map $M$, maintaining \textbf{Strong convergence}.
\end{proof}

\subsection{Outlook on Systemic Access Control}
\label{sec:ac-outlook}
Access control is usually enforced by a centralized entity in a strongly consistent way.
In a distributed, weakly-consistent setting, we have to embrace that time is only a partial ordering, events and administrative changes happen concurrently, and there is no such thing as consensus on a total order of events or on which policies are in in effect “now”~\cite{weber2016access}.

Previous works on access control for CRDTs mainly focus on filesystem-like cryptographically-enforceable access control in closed groups~\cite{rault2022distributed,yanakieva2021access},
while our outlook is inspired by the granular authorizations and administrative permissions of  Matrix~\cite{matrix-spec} and similar systems~\cite{lf-auth}.

The idea is to provide {\em systemic} access control, i.e.,
storing attributes needed for policies as well as policies itself in the types, thereby allowing integration of access decisions in the CRDT functions and gaining a decentralized enforcement.
While the bare EDP formalizations from~\cref{sec:edp-spec} do not require authenticity of updates, decentralized enforcement requires that a replica can prove to another replica that a third replica was responsible for an update, and thereby requires transferable authentication through digital signatures.
Such a systemic access control can be constructed by a composition of a CEDP -- to gain a concept of logical time through the happened-after relation -- and multiple EPMs.

The main challenge is to treat concurrent administrative changes securely.
To deal with concurrent, conflicting changes, the key idea is that a concurrent update is never \emph{rejected} if it was authorized for its downward closure, it just might be ignored on linearization if another change wins.
A prototype model of such a composed data type
is defined and discussed in~\cref{sec:acedpm}, together with a proof sketch of the CRDT property.

\section{Conclusion \& Future Work}
Based on our mathematical formalization around partially-ordered sets, we presented Extend-only Directed Posets (EDPs) as
unifying generalization of ideas around DAG CRDTs in Byzantine environments. 
The state-based formalization shows the essence of these CRDTs, while the operation-based formalization makes the optimization for efficiency explicit.
Based on EDPs, we can derive modular building blocks and compose complex Byzantine-tolerant CRDTs with ease, as exemplified by the map data type and indicated by the use case for systemic decentralized access control.
In composition, these building blocks may be the first capable of an access-control-included formalization of the CRDT-based Matrix messaging system that aspires among public sectors~\cite{matrix-holiday-2022}.

This work-in-progress provides a step towards formal verification of strong eventual consistency in Byzantine environments of both EDP-based system designs as well as corresponding implementations (like Matrix), including security properties of decentralized access control.

\newpage{}

\printbibliography{}

\appendix

\clearpage{}
\newpage{}
\section{Integrating Systemic Access Control in a CEDP/EPM Composition}
\label{sec:acedpm}


In this appendix, we compose a CEDP and one or more EPMs to provide access control for the corresponding replicated data types.
Attributes needed for policy evaluation as well as the policies themselves are stored in the replicated data types and serve as basis for access decision and enforcement. Furthermore, also administrative access control tasks like changing rights and policies are governed by the attributes and policies within the replicated data types. Therefore, one does not need to resort to an external system for administering access control.
We refer to this approach as {\em systemic} access control and call
the composed data type {\it Access-Controlled Extend-only Directed Poset and Maps} (ACEDPM).
The approach is inspired by the way the Matrix messaging platform handles access control.

Access control is defined in terms of subjects that perform different actions on objects.
In our messaging example, a user subject performs actions like reading and writing chat messages, or querying and changing group membership, on objects like other users or a chat group.
An event is the act of a subject performing an action on an object.
We focus on consistency, and as read actions cannot harm consistency, we restrict ourselves to write actions in the following.
We store (write) events as set elements in the data type.
To gain a concept of logical time, we store the happened-after-equal relation among stored events, i.e., a CEDP.
One-off events like chat messages are stored directly in the CEDP, while change events of attributes like group membership are stored via one or more EPMs on top of the CEDP.

Whoever initiates the CEDP and EPM replicated objects also defines “the rules” of the composed ACEDPM object, i.e., the initial attributes and policies as foundation on which all further regular and administrative events have to be authorized.
Classical access control thinking assumes a centralized, strongly-consistent access control enforcement based on data and policy changes ordered by linear, real-world time.
The challenge here is that CEDP and EPM only provide strong eventual consistency (SEC) as trade-off for availability under partition, and we want to integrate access control while keeping that trade-off. 
In decentralized, weakly-consistent access control, different replicas do not necessarily make the same decision at one point of real-world time, as data is not necessarily consistent and changes are ordered by partial, logical time.
However, if the composed ACEDPM is still a Byzantine-tolerant CRDT, i.e., can still provide SEC, its state and access control is still secure in the sense that Byzantine replicas cannot force correct replicas to diverge in their state and access decisions permanently.
The challenge is to treat concurrent administrative events securely.
The ACEDPM keeps all concurrent events and linearizes them.
We explain the approach in \cref{sec:acedpm-spec}, and show a proof sketch for SEC in~\cref{sec:acedpm-sec}.

\subsection{ACEDPM Specification}
\label{sec:acedpm-spec}
An ACEDPM deals with concurrent events in a CRDT-usual way: it accepts \emph{all} valid events, and then breaks ties among concurrent events, here using an access-control-based linear ordering among participants.
We assume a peer-to-peer system architecture where one replica corresponds to a single user with full control.
An ACEDPM object stores events that consist of an action $\mathvar{act}$ of subject $\mathvar{sbj}$ (the sending replica) and the action's content $\mathvar{cnt}$.
All events have the ACEDPM replicated object as implicit object, and may have an explicit object $\mathvar{obj}$, e.g., another user.
We represent events as tuples $x = (\mathvar{act}, \mathvar{sbj}, \mathvar{cnt})$, or as $x = (\mathvar{act}, \mathvar{sbj}, \mathvar{obj} \mapsto \mathvar{cnt})$ if an explicit object is present.
A few examples:
Alice writing a text message is $(\mathval{chat}, \mathval{Alice}, \mathval{"Hi!"})$.
If Alice changes group membership to include Bob, we get $(\mathval{membership}, \mathval{Alice}, \mathval{Bob} \mapsto IN)$.
We get $(\mathval{level}, \mathval{Alice}, \mathval{Alice} \mapsto 2342)$ if Alice sets her own access level to 2342.
As replicas need to verify the correctness of the subject identifier $\mathvar{sbj}$ and be able to prove its correctness to others, events require digital signatures.

In the following, we use two EPMs on top of the CEDP causal event set to integrate access control: one map $M$ for a group membership attribute and one map $L$ for an access level attribute.
The first map $M$ is accessed via a \emph{getM} variant of \emph{get} query function from \cref{alg:epmrdt}.
Possible membership values are for example $\{OUT, IN, INVITE, BAN\}$.
For authorization, we will require a subject to be $IN$, but e.g. Matrix also constrains membership state transitions.
The second map $L$ is used for an access level attribute of subjects and actions, accessed via \emph{getL}.
Levels are integers associated with subjects and actions, and allow a subject to store events for all actions with less or equal level.
In addition, levels impose an access-control-based linear ordering on subjects and objects, as well as their respectively sent events.
For more details on a possible access control system, please cf.~Level- and Attribute-based Access Control as defined in~\cite{JacobBeckerGrashoefer2020_1000120022}.

For systemic access control, we change the EDP and EPM models to use an authorization function $\mathfn{authorized}(u, T)$ for update $u \in \mathbb{R}$ and point in logical time $T \in \mathcal{P}(\mathbb{R})$ before applying updates, and an access-control-based priority relation $R_a$ to break ties.
We first explain why and where changes have to be made in the models, and then show algorithms for $\mathfn{authorized}(u, T)$ and $R_a$.
In the \emph{effect} function of the EDP (c.f.~\cref{alg:edprdt}), we have to assert $\mathfn{authorized}(u, \mlb(u))$ to ignore received updates that are not authorized by their maximal lower bounds.
In the \emph{get} function of the EPM (c.f.~\cref{alg:epmrdt}), a linearized update $T_i$ also must only be applied if $\mathfn{authorized}(T_i, \bigcup_{j=0}^{j<i} T_{j})$, to ignore stored updates that are not authorized by updates preceding the linearization.
Finally, the priority relation $R_a$ has to be used instead of the hash-based ordering $R_h$ in the \emph{linearize} function of the EPM, as $R_h$ orders updates deterministic but arbitrarily.
The strict linear ordering $R_a$ instead prioritizes revocations, and all actions according to their subject's level, before falling back to $R_h$.

\begin{algorithm}
\caption{Level- and Attribute-Based Authorization}
\label{alg:leabac-authorization}
\begin{algorithmic}
\Fn{authorized}{$u \in \mathbb{R}, T \in \mathcal{P}(\mathbb{R})$}{$a \in {0, 1}$}
  \State \Comment{whether $u$ is authorized at point in logical time $T$}
  \State $x \gets \max(X(u)), M \gets \mathfn{getM}(T), L \gets \mathfn{getL}(T)$
  \State $\mathvar{group}_a \gets M[x.\mathvar{sbj}] = IN$
  \State $\mathvar{action}_a \gets L[x.\mathvar{act}] \leq L[x.\mathvar{sbj}]$
  \State $\mathvar{object}_a \gets x.\mathvar{sbj} = x.\mathvar{obj} \lor L[x.\mathvar{obj}] < L[x.\mathvar{sbj}]$
  \State $\mathvar{level\_cap} \gets x.\mathvar{act} = \mathval{level} \Rightarrow x.\mathvar{cnt} \leq L[x.\mathvar{sbj]}$
  \State \Comment{For "set level" actions, cap new level by subject level}
  \State $a \gets \mathvar{group}_a \land \mathvar{action}_a \land \mathvar{object}_a \land \mathvar{level\_cap}$
\EndFn
\end{algorithmic}
\end{algorithm}

The \emph{authorized} function shown in \cref{alg:leabac-authorization} decides whether an upward extension $u$ is authorized
by the attributes at the logical point in time $T$.
Event $x$ is authorized if the subject is a) member of the group, b) authorized for the action, c) authorized for acting on the (optional) object, and d) does not exceed their access level.
The subject needs a level strictly greater than the object to prevent conflicting permission revocations to create inconsistencies: Same-level subjects cannot revoke each others' permissions.

The access-control-based priority relation $R_a$ shown in \cref{alg:acedpm-priority} is a strict linear ordering among all upward extensions $\in \mathbb{R}$.
In first order, $R_a$ prioritizes revocation actions over other actions that cannot reduce permissions, to ensure that revocations are applied before other concurrent upward extensions.
In second order, for actions that are either both revocations or both other actions, it prioritizes actions according to their subject's level, to ensure that actions of subjects with higher level are applied before.
In third order, for actions that are from subjects of equal level, $R_a$ falls back to deterministic but arbitrary hash-based ordering.

\begin{algorithm}
\caption{Access-Control-Based Priority Relation $R_a$}
\label{alg:acedpm-priority}
\begin{algorithmic}
\State $R_a \gets \{ (u_1, u_2) \in \mathbb{R}^2 | \mathfn{prior}_a(u_1, u_2) \}$
\Fn{$\mathfn{prior}_a$}{$u_1, u_2 \in \mathbb{R}$}{$b \in \{0, 1\}$}
  \State \Comment{whether $u_1$ is prior to $u_2$ regarding access control}
  \State $x_1 = \max(X(u_1)), x_2 = \max(X(u_2))$
  \State $L_1 \gets \mathfn{getL}(\mlb(u_1)), L_2 \gets \mathfn{getL}(\mlb(u_2))$
  \State $b \gets \mathfn{rvc}(u_1) \land \neg \mathfn{rvc}(u_2)$ \Comment{prioritize revocations}
  \If{$\mathfn{rvc}(u_1) = \mathfn{rvc}(u_2)$}
    \State $b \gets L_1[x_1.\mathvar{sbj}] < L_2[x_2.\mathvar{sbj}])$
    \State \Comment{if both / neither is revocation, order by level}
    \If{$L_1[x_1.\mathvar{sbj}] = L_2[x_2.\mathvar{sbj}]$}
      \State $b \gets h(u_1) < h(u_2)$
      \State \Comment{if equal subject level, order by hash}
    \EndIf
  \EndIf
\EndFn
\Statex
\Fn{rvc}{$u \in \mathbb{R}$}{$b \in {0, 1}$} \Comment{whether $u$ is revocation}
  \State $x \gets \max(X(u))$
  \State $M_\mathvar{pre} \gets \mathfn{getM}(\mlb(u)), M_\mathvar{post} \gets M_\mathvar{pre} \uplus u$
  \State $L_\mathvar{pre} \gets \mathfn{getL}(\mlb(u)), \quad L_\mathvar{post} \gets M_\mathvar{pre} \uplus u$
  \State
  $b \gets (M_\mathvar{pre}[x.\mathvar{obj}] = IN \land M_\mathvar{post}[x.\mathvar{obj}] \neq IN)$
  \State $\hspace*{2em} \lor L_\mathvar{pre}[x.\mathvar{obj}] > L_\mathvar{post}[x.\mathvar{obj}]$
\EndFn
\end{algorithmic}
\end{algorithm}

By prioritizing revocations in the linearization of map updates, revocations gain “concurrent+causal for-each” semantics, a concept Weidner identified\footnotemark{} based on~\cite{weidner2020composing}: a revocation acts against causally succeeding as well as concurrent map updates.
When using the linearization of all elements as query function, revocations gain these desired semantics for all ACEDPM updates, i.e., for both CEDP and EPMs.
\footnotetext{\url{https://mattweidner.com/2022/02/10/collaborative-data-design.html}}

\subsection{Byzantine Strong Eventual Consistency of ACEDPM}
\label{sec:acedpm-sec}

One might wonder whether concurrent revocations against each other could break eventual delivery or strong convergence.
However, SEC is ensured because a concurrent update is never \emph{rejected} if it was authorized by its downward closure, it just might be ignored in the linearization phase in case the other revocation wins.
Therefore, all correct replica states will eventually contain both revocations and their downward closure, regardless of reception and application order.
Linearization is only based on the relation of revocations to each other and their respective downward closure, whereby all correct replicas will eventually reach the same decision.

\begin{theorem}
  Under the assumption of a connected component of all correct replicas and eventual communication among the component,
  the ACEDPM is an op-based Conflict-free Replicated Data Type
  even in face of an arbitrary number of Byzantine replicas.
\end{theorem}
\begin{proof}
Based on SEC of the underlying CEDP/EPM.
\textbf{Self-update:}
A correct replica only generates updates that pass its local, current authorization checks.
Thereby, authorization does not hinder the replica from immediate self-update.
\textbf{Eventual update:} An update generated by a correct replica is authorized based on its maximal lower bounds.
Receiving correct replicas check whether an update is authorized only based on its maximal lower bounds.
Thereby, they only filter out invalid updates that cannot have been sent by a correct replica.
But they apply all updates authorized by their maximal lower bounds, even if a concurrent update has already revoked the permissions in the EPM state based on the receiving replica state's maximal elements.
Resolving conflicts is moved towards linearization, and not solved through not applying updates eventually.
\textbf{Strong convergence:} For the CEDP, the authorization is checked before applying an update, whereby authorization cannot make two replicas differ that have applied the same updates.
For the EPMs, there is an authorization check inside the linearization that ignores updates if they are not authorized in linearization order.
But, linearization is deterministic and independent from the order in which updates were applied.
Therefore, replicas in the same CEDP state come to the same linearization, and thereby also to the same EPM states.
\end{proof}

\end{document}